\documentclass{gretsi}

\usepackage{cite}
\usepackage{amsmath,amssymb,amsfonts}
\usepackage{algorithmic}
\usepackage{graphicx}
\usepackage{textcomp}
\def\BibTeX{{\rm B\kern-.05em{\sc i\kern-.025em b}\kern-.08em
		T\kern-.1667em\lower.7ex\hbox{E}\kern-.125emX}}

\usepackage{url}

\usepackage{pgf}
\usepackage{bm}
\let\labelindent\relax
\usepackage{enumitem}

\usepackage{amsthm}
\newtheorem{definition}{Definition}

\newtheorem{lemma}{Lemma}
\newtheorem{conjecture}{Conjecture}
\newtheorem{theorem}{Theorem}
\newtheorem{corollary}{Corollary}

\usepackage{graphicx, color, caption, subcaption}
\graphicspath{{fig/}}

\newcommand{\fakecite}[1]{\textcolor{red}{[X]}}

\newcommand{\guillemets}[1]{``#1''}

\newcommand{\diag}{\mbox{diag}}

\newcommand{\set}[1]{\left\{#1\right\}}

\newcommand{\myvector}[1]{{#1}}

\newcommand{\mymatrix}[1]{\bm{#1}}

\newcommand{\R}{\mathbb{R}}
\newcommand{\K}{\mathbb{K}}
\renewcommand{\L}{\mathbb{L}}
\newcommand{\Q}{\mathbb{Q}}

\newcommand{\rF}{\mathcal{F}}
\newcommand{\rG}{\mathcal{G}}

\newcommand{\rP}{\boldsymbol{\mathcal{P}}}

\usepackage{stmaryrd}

\newcommand{\norm}[1]{\left\lVert#1\right\rVert}
\newcommand{\abs}[1]{\left|#1\right|}

\newcommand{\ie}{{i.e.},}
\newcommand{\eg}{{e.g.},}

\DeclareMathOperator*{\rank}{rank}

\DeclareUnicodeCharacter{2212}{-}

\begin{document}
	\title{Pseudorange Rigidity and Solvability of Cooperative GNSS Positioning}
	
	\auteur{\coord{Colin}{Cros}{1,2},
		\coord{Christophe}{Prieur}{1},
		\coord{Pierre-Olivier}{Amblard}{1},
		\coord{Jean-François}{Da Rocha}{2}}
	
	\adresse{\affil{1}{Univ. Grenoble Alpes, CNRS, Grenoble-INP, GIPSA-Lab, F-38000 Grenoble}
		\affil{2}{Telespazio FRANCE, F-31100 Toulouse}}

	\email{colin.cros@gipsa-lab.fr,
		christophe.prieur@gipsa-lab.fr\\
		pierre-olivier.amblard@cnrs.fr, jeanfrancois.darocha@telespazio.com}

	\resumeanglais{Global Navigation Satellite Systems (GNSS) are a widely used technology for positioning and navigation. GNSS positioning relies on pseudorange measurements from satellites to receivers. A pseudorange is the apparent distance between two agents deduced from the time-of-flight of a signal sent from one agent to the other. Because of the lack of synchronization between the agents' clocks, it is a biased version of their distance. This paper introduces a new rigidity theory adapted to pseudorange measurements. The peculiarity of pseudoranges is that they are asymmetrical measurements. Therefore, unlike other usual rigidities, the graphs of pseudorange frameworks are directed. In this paper, pseudorange rigidity is proved to be a generic property of the underlying undirected graph of constraints. The main result is a characterization of rigid pseudorange graphs as combinations of rigid distance graphs and connected graphs. This new theory is adapted for GNSS. It provides new insights into the minimum number of satellites needed to locate a receiver, and is applied to the localization of GNSS cooperative networks of receivers. The interests of asymmetrical constraints in the context of formation control are also discussed.}
	
	\maketitle

	\section{Introduction}\label{sec: Introduction}
	
	Global Navigation Satellite Systems (GNSS) provide an effective and low-cost solution for localization. They rely on constellations of satellites, equipped with highly accurate atomic clocks. In these systems, the satellites broadcast signals that contain information about the precise location of the emitting satellite as well as the time of emission of the signal transmitted \cite{kaplan2017understanding}. These signals are received by GNSS receivers on the ground. The receivers measure the time of reception and deduce the times-of-flight of the signals. These times are converted into distances by multiplying by the signal's celerity. As the receivers are generally not synchronized with the satellites, the distance obtained is a biased version of the distance between the satellite and the receiver. It is called a \emph{pseudorange}. The bias comes from the delay between the satellite's clock and the receiver's clock that is also multiplied by the signal's celerity. This bias is the same for every satellite within a GNSS constellation and must be estimated. Indeed, a delay of $10$ns would produce a range error of about $3$m. By receiving signals from multiple satellites, a receiver can determine its position on the Earth's surface by solving the nonlinear system of equations induced by the pseudoranges. Therefore, GNSS positioning is a multilateration problem, similar to other systems that were used before it, such as the \emph{Long Range Navigation} (LORAN) systems \cite{frank1983current}.
	
	\begin{figure}
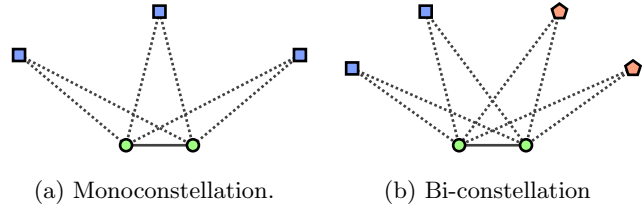

		\begin{subfigure}{.49\linewidth}
			\centering
			{\input{fig/bi_receivers_mono_constellation_measurements.pgf}}
			\caption{Monoconstellation.}
			\label{sfig: Two receivers examples mono constellation}
		\end{subfigure}
		\begin{subfigure}{.49\linewidth}
			\centering
			{\input{fig/bi_receivers_bi_constellation_measurements.pgf}}
			\caption{Bi-constellation}
			\label{sfig: Two receivers examples bi constellation}
		\end{subfigure}
		\caption{Graph of measurements of cooperative networks. Each network is composed of two GNSS receivers represented by circles. One constellation of satellites is represented by squares and another by pentagons. The dotted lines represent pseudorange measurements, and the solid lines inter-receiver distance measurements.}
		\label{fig: Two receivers examples}
	\end{figure}
	
	The minimum number of pseudorange measurements required to locate a receiver is $3 + C$ where $C$ is the number of GNSS constellations used. The usual justification is that there are $3+C$ unknowns in the system: $3$ for the receiver position, plus $1$ per GNSS constellation clock bias. Each pseudorange equation is used to solve for one unknown, and therefore the localization problem is solvable with $3+C$ pseudoranges. The recent development of network systems raises the question of their cooperative positioning. When a node is unable to use GNSS (completely or partially), it can cooperate with the other nodes in the network to estimate its position, \eg{} by measuring distances with its neighbors and performing trilateration. Collaborative positioning has been proved to be an efficient solution to improve the precision \cite{minetto2019trade} and to extend the availability \cite{causa2018multi} of the GNSS. However, the localizability of a network of GNSS receivers from a set of given pseudorange measurements has never been answered. In general, the minimal number of pseudorange measurements required for locating a network is unknown. For example, consider a pair of receivers measuring pseudoranges from only $2+C$ satellites from $C$ different GNSS constellations. They cannot estimate their positions as it would require one additional measurement. Assume they also measure the distance between them, as illustrated in Fig.~\ref{fig: Two receivers examples}. With this additional measurement, can they estimate their positions? The aim of this paper is to answer this question for a general class of cooperating networks. The solvability of such localization problems is intrinsically linked to the notion of rigidity.
	For example, global rigidity ensures that the problem has a unique solution as detailed for distance measurements in \cite{aspnes2006theory}, or for angle measurements in \cite{fang2020angle}. Nevertheless, in many applications (local) rigidity is sufficient to solve the problem.
	Rigidity is the capacity of a structure to maintain its shape by solely maintaining some constraints between pairs of nodes (called agents in the sequel). In general, it implies that the Jacobian matrix of the measurements has a maximal rank \cite{asimow1979rigidity} which is sufficient to design estimation algorithms. Rigidity was first introduced with distance constraints to study the stability of bar-and-joint structures, see \eg{} \cite{jackson2007notes} for an introduction. Since then, it has been adapted to other forms of interactions between the agents, \eg{} bearing rigidity \cite{zhao2015bearing} and angle rigidity \cite{chen2020angle, fang2020angle} have been derived for agents equipped with cameras providing the angles or the bearings between the agents. Recently, Joint Position-Clock (JPC) rigidity  \cite{wen2022clock} has been introduced to treat symmetrical time-of-flight measurements. Unfortunately, in GNSS, the signals are sent only from satellites to receivers, which results in asymmetrical measurements. Most of the rigidities introduced in the literature are based on symmetrical measurements and symmetrical constraints. Consequently, the frameworks are described by undirected graphs of constraints. Pseudoranges are by nature asymmetrical, and therefore they cannot be treated with the existing rigidity theories.

	In this paper, we focus on sensor networks performing pseudorange measurements and distance measurements. We answer the question of the localizability of the network given the set of measurements, \ie{} we answer the following question. Is the positions of the agents identifiable given the measurements? To answer this question, we introduce and characterize a new form of asymmetrical rigidity adapted to the pseudorange context. The main contributions are the following:
		\begin{enumerate}
			\item Pseudorange rigidity is introduced. The main difference with the usual rigidities is that the graphs of pseudorange frameworks are directed to account for the asymmetry of the measurements. We prove that pseudorange rigidity is a generic property of the underlying undirected graph. Furthermore, we prove that the rank of the pseudorange rigidity matrix can be expressed by separating the spatial variables from the clock parameters. The consequence is that a pseudorange graph is rigid if and only if it is the combination of a distance rigid graph and a connected graph.
			\item Pseudorange rigidity is extended to define the rigidity of GNSS networks. GNSS rigidity brings a new justification to the common wisdom about the minimum number of satellites required to locate a single receiver. Furthermore, it applies to cooperative networks of GNSS receivers. It helps to identify the connections required to locate a network and to design localization algorithms.
			\item The interests of pseudorange rigidity for formation control are presented. To preserve the $2$D formation of a group of agents, the pseudorange point of view allows to reduce by up to $25$\% the number of measurements, with respect to a classical two-way ranging method. For $3$D formations, this number is reduced by up to $33$\%.
			\item New algebraic concepts are exposed to isolate spatial variables from the other variables in rigidity matrices. The technique employed in this paper may be reused for other types of rigidity, \eg{} JPC rigidity \cite{wen2022clock}.
	\end{enumerate}
	
	The rest of this paper is organized as follows. Section~\ref{sec: Preliminaries} introduces the notion of pseudorange and provides some background on rigidity. Pseudorange rigidity is introduced in Section~\ref{sec: Pseudorange rigidity}. The rigidity of pseudorange graphs is characterized in Section~\ref{sec: Characterization of generic pseudorange rigidity}. Section~\ref{sec: GNSS rigidity} adapts these new results to the GNSS context and Section~\ref{sec: Pseudorange rigidity for formation control} discusses the applications of pseudorange rigidity for formation control. Finally, Section~\ref{sec: Concluding remarks} gives some perspectives.
	
	\medbreak
	\textbf{Notation.}
	In the sequel, matrices are denoted in uppercase boldface variables \eg{} $\mymatrix{M} \in \R^{n\times m}$, and the Euclidean norm of a vector is denoted as $\norm{\myvector{x}}$. The dimension of the space in which the agents live is denoted as $d$, and it is assumed set. In practice $d = 3$, but the results presented here are valid for any $d \ge 2$. A simple graph with a vertex set $V$ and an edge set $E$ is denoted as $G=(V,E)$. Undirected simple graphs are named with Latin letters, \eg{} $G$, while directed simple graphs are named with Greek letters, \eg{} $\Gamma$. A directed edge is called an arc. For a simple directed graph $\Gamma = (V,E)$, $\tilde \Gamma = (V, \tilde E)$ denotes the undirected multi-graph induced by $\Gamma$ where $\tilde E$ denotes the multiset of the edges. An edge can appear 0 times once or twice in the $\tilde E$. For a general background on graph definitions and properties (incidence matrix, connectivity, cycles, etc.), we refer to \cite{bollobas1998modern}. The cardinality of a set $A$ is denoted as $\abs{A}$.

	\section{Preliminaries}\label{sec: Preliminaries}

	\subsection{Pseudorange measurements}
	
	In the sequel, \guillemets{agent} is a generic term referring to satellites and receivers.
	In the context of GNSS, a pseudorange is an indirect measurement of a distance between two agents based on the time-of-flight of a signal. As the time of emission and the time of reception of the signal are measured by two different and potentially unsynchronized clocks, it is a biased version of the distance. The time $t^i$ given by the clock of the $i$th agent is modeled as:
	\begin{equation}\label{eq: Model of clock 1}
		t^i = t + \tau_i,
	\end{equation}
	where $t$ is some virtual absolute time and $\tau_i$ is the bias with respect to that time. The position of the $i$th agent is denoted as $\myvector{x}_i \in \R^d$. Its clock bias is taken premultiplied by the signal's celerity $c$ to be homogeneous to a length, and it is denoted as $\beta_i \triangleq c\tau_i$. The pseudorange from an agent $u$ to another $v$ is:
	\begin{equation}
		\rho_{uv}\triangleq c\left(t_{r(uv)}^v - t_{e(uv)}^u\right) = \norm{\myvector{x}_u - \myvector{x}_v} + \beta_v - \beta_u
	\end{equation}
	where $t_{e(uv)}$ and $t_{r(uv)}$ denote respectively the time of emission and the time of reception of a signal sent from $u$ to $v$. Note that the pseudorange equals the distance if and only if $\beta_u = \beta_v$, \ie{} if the agents' clocks are synchronized.
	
	\subsection{Some definitions for rigidity analysis}\label{ssection: Generalities on rigidity}
	Consider a network composed of $n$ agents in $\R^d$. The $i$th agent is characterized by a vector $\myvector{p}_i$. Usually $\myvector{p}_i$ is the position $\myvector{x}_i \in \R^d$ of the agent, but $\myvector{p}_i$ can contain other parameters such as clock parameters. The \textbf{configuration} of the network is the vector $\myvector{p} = \begin{bmatrix}	\myvector{p}_1^\intercal & \dots & \myvector{p}_n^\intercal 
	\end{bmatrix}^\intercal$. The interactions between the agents are represented in a graph $G =(V,E)$. The edges in $E$ represent constraints between the agents in $V$, \eg{} distance constraints or bearing constraints. The number of constraints is $m = \abs{E}$. The pair $(G, \myvector{p})$ is called a \textbf{framework}. The evaluation of the constraints is the vector $\myvector{F}_X(G, \myvector{p})= \begin{pmatrix} \dots & f_{X, uv}(\myvector{p})^{\intercal} &\dots \end{pmatrix}^\intercal \in \R^{m k_X}$, where in this definition $X$ denotes the type of rigidity, $f_{X, uv}(\myvector{p}) \in \R^{k_X}$ is the evaluation of the constraint induced by the edge $uv \in E$, and the edges are assumed to be ordered.
	
	Two frameworks $(G, \myvector{p})$ and $(G, \myvector{p}')$ are said to be \textbf{equivalent} if $\myvector{F}_X(G, \myvector{p}) = \myvector{F}_X(G, \myvector{p}')$. They are said to be \textbf{congruent} if $\myvector{F}_X(K, \myvector{p}) = \myvector{F}_X(K,\myvector{p}')$, where $K$ is the complete graph whose edge set is $E_K = \set{uv \in V^2, u < v}$.
	
	A framework $(G, \myvector{p})$ is \textbf{rigid} if there exists $\epsilon > 0$ such that for all $\myvector{p}'$ satisfying $\norm{\myvector{p} - \myvector{p}'} < \epsilon$, $(G, \myvector{p})$ and $(G, \myvector{p}')$ equivalent implies $(G, \myvector{p})$ and $(G, \myvector{p}')$ congruent. A non-rigid framework is called \textbf{flexible}. A framework $(G, \myvector{p})$ is \textbf{globally rigid} if for all $\myvector{p}'$, $(G, \myvector{p})$ and $(G,\myvector{p}')$ equivalent implies $(G, \myvector{p})$ and $(G, \myvector{p}')$ congruent.
	
	The \textbf{rigidity matrix} \cite{jackson2007notes} of a framework is defined as the Jacobian of the evaluation function:
		\begin{equation}
			\mymatrix{R}_X(G, \myvector{p}) \triangleq \frac{\partial \myvector{F}_X(G, \myvector{p})}{\partial \myvector{p}}.
	\end{equation}
	
	A \textbf{velocity vector} \cite{jackson2007notes} $\myvector{q}$ is a variation of $\myvector{p}$. If $\mymatrix{R}_X(G, \myvector{p})\myvector{q} = \myvector{0}$, $\myvector{q}$ is said to be \textbf{admissible} for the framework. A velocity vector admissible for all frameworks is called \textbf{trivial}. A framework $(G, \myvector{p})$ is \textbf{infinitesimally rigid} if all its admissible velocity vectors are trivial.
	
	\subsection{Different notions of rigidity}\label{ssec: Different rigidities}
	
	This section provides a summary of different forms of rigidity. It does not intend to be exhaustive but focuses on the rigidities related with pseudorange rigidity introduced in the next section.
	
	\subsubsection{Distance rigidity}
	
	Distance rigidity is the original and most studied rigidity. An agent is represented by its spatial coordinates $\myvector{p}_i = \myvector{x}_i \in \R^d$ and the edges of $G$ constrain the distances between pairs of agents. The constraint induced by an edge $uv$ is $f_{D, uv}(\myvector{p}) \triangleq \norm{\myvector{x}_u - \myvector{x}_v} = \delta_{uv}$ where $\delta_{uv}$ is a given distance. The trivial motions correspond to translations and rotations of the framework.
	
	\subsubsection{Distance rigidity in elliptical and hyperbolic space}
	
	Distance rigidity has been extended to non-Euclidean spaces, see \eg{} \cite{saliola2007some, schulze2012coning}. An interesting case for our study is the Minkowski hyperbolic space. In this case, the agents are parameterized by $\myvector{p}_i = \begin{pmatrix} \myvector{x}_i^\intercal & \gamma_i \end{pmatrix}^\intercal \in \R^{d+1}$, where $\gamma_i$ is a scalar. The constraint induced by an edge is $f_{M, uv}(\myvector{p}) \triangleq \norm{\myvector{x}_u - \myvector{x}_v}^2 - (\gamma_u - \gamma_v)^2 = \delta_{uv}$ where $\delta_{uv}$ is a given \guillemets{distance} (that may be negative). Unlike a pseudorange, the Minkowski distance $f_{M, uv}(\myvector{p})$ is symmetrical in $u$ and $v$.
	
	\subsubsection{Bearing rigidity}
	
	Bearing rigidity focuses on preserving the shape of a framework by constraining the bearings between the agents, see \eg{} \cite{zhao2019bearing} for an overview. The $i$th agent is also represented by its spatial coordinates $\myvector{p}_i = \myvector{x}_i \in \R^d$. The constraint induced by an edge $uv$ of the graph is $f_{B, uv}(\myvector{p}) \triangleq \frac{\myvector{x}_u - \myvector{x}_v}{\norm{\myvector{x}_u - \myvector{x}_v}} = \mymatrix{\alpha}_{uv} \in \R^d$ where $\mymatrix{\alpha}_{uv}$ is a given bearing. The trivial motions correspond to translations and stretching of the framework.
	
	\subsubsection{Clock rigidity}
	
	Clock rigidity was recently introduced in \cite{wen2022clock}. It focuses on preserving the synchronization between the clocks of the agents. The clock model considered in clock rigidity has two parameters. The time $t^i$ of the $i$th agent's clock is modeled as:
		\begin{align}\label{eq: Model of clock 2}
			t^i &= w_i t + \tau_i, &\Leftrightarrow&& t &= \alpha_i t^i + \lambda_i,
		\end{align}
		where $w_i$ and $\alpha_i$ are clock skews, and $\tau_i$ and $\lambda_i$ are clock biases. The $i$th agent's clock is parameterized by $\myvector{p}_i = (\alpha_i, \lambda_i) \in \R^2$. The authors assumed that times-of-flight are always measured in both directions \cite[Assumption 1]{wen2022clock}. Under this assumption, each pair of measurements imposes:
		\begin{subequations}\label{eq: Position-Clock constraints}
			\begin{align}
				\norm{\myvector{x}_u - \myvector{x}_v} &= c(\alpha_v t^v_{r(uv)} + \lambda_v - \alpha_u t^u_{e(uv)} - \lambda_u), \\
				\norm{\myvector{x}_u - \myvector{x}_v} &= c(\alpha_u t^u_{r(vu)} + \lambda_u - \alpha_v t^v_{e(vu)} - \lambda_v).
			\end{align}
		\end{subequations}
		Therefore, an edge $uv$ of $G$ induces the following constraint on the clock parameters: $f_{C, uv}(\myvector{p}) \triangleq \alpha_u \bar T^u + \lambda_u - \alpha_v \bar T^v - \lambda_v = 0$ where $\bar T^u$ and $\bar T^v$ are constant. The authors proved that clock rigidity is strongly connected with bearing rigidity in $2$D \cite{wen2022clock}.
	
	\subsubsection{Joint Position-Clock rigidity}
	JPC rigidity was also introduced in \cite{wen2022clock} to preserve both the distances and the clock synchronizations between the agents. In this case, $\myvector{p}_i = \begin{pmatrix}\myvector{x}_i^\intercal & \alpha_i & \lambda_i \end{pmatrix}^\intercal \in \R^{d+2}$ is the combination of the spatial coordinates and the clock parameters of the agents. The edges of $G$ constrain the two times-of-flight between the agents (in both directions). An edge $uv \in G$ constrains the two constraints \eqref{eq: Position-Clock constraints}. To study JPC rigidity, the authors of \cite{wen2022clock} used the bi-directional assumption to separate the distance constraint from the clock constraint.
	
	\subsubsection{Coordinated rigidity}
	Coordinated rigidity was introduced in \cite{nixon2018rigidity, schulze2022frameworks}. It is an extension of distance rigidity that allows a new type of movements. The edges are split into groups and the edges among a group can expand simultaneously. Formally, the graph is enriched with a coloration of its edges $c = (E_0, \dots, E_k)$, and the configuration is enriched with a vector $r \in \R^k$ associating to each color a bias. Two frameworks $(G, c, \myvector{p}, \myvector{r})$ and $(G, c, \myvector{p}', \myvector{r}')$ are said to be equivalent if:
		\begin{align*}
			\norm{ \myvector{p}_u - \myvector{p}_v} &= \norm{ \myvector{p}'_u - \myvector{p}'_v} & \forall uv&\in E_0, \\
			\norm{ \myvector{p}_u - \myvector{p}_v} + r_l &= \norm{ \myvector{p}'_u - \myvector{p}'_v} + r'_l & \forall uv&\in E_l, \,  1 \le l \le k.
		\end{align*} 
		This context fits for GNSS: each color corresponds to a pair receiver-constellation and each bias to the offset of the receiver's clocks on the constellation time. The pseudorange point of view is however more general, \eg{} it can be applied if the receivers cooperate by pseudorange measurements.
	
	\subsection{About the lack of asymmetrical rigidity}
	
	All the constraints presented Section~\ref{ssec: Different rigidities} are symmetrical: there is no difference between constraining the pair $(u,v)$ or the pair $(v,u)$. For clock rigidity and JPC rigidity, even if the measurements are not symmetrical, the symmetry comes from the fact that the times-of-flight are measured in both directions. The specificity of our study lies in the asymmetry of the pseudoranges. In the GNSS context, we cannot assume that the pseudoranges are measured in both directions as they are taken only from satellites to receivers. This motivates the need to study pseudorange rigidity as a new concept. The important difference with the literature is that pseudorange graphs are directed to account for the asymmetry of the measurements. JPC rigidity can also be considered in an asymmetrical context. In Section~\ref{sec: Pseudorange rigidity for formation control}, we discuss this extension and the main difference with pseudorange rigidity.

	\section{Pseudorange rigidity}\label{sec: Pseudorange rigidity}
	
	\subsection{Pseudorange frameworks}
	
	Consider a network of $n$ agents, the $i$th agent is parameterized by $\myvector{p}_i = \begin{pmatrix} \myvector{x}_i^\intercal & \beta_i \end{pmatrix}^\intercal \in \R^{d+1}$ where $\myvector{x}_i$ is its position and $\beta_i$ its clock bias. Consider also a directed graph $\Gamma = (V, E)$ representing the pseudorange measurements between the agents. The set $E$ is a set of arcs (directed edges). To maintain generality, no restriction on the topology of the graph is imposed. For a pair of vertices $(u,v) \in V^2$, $E$ can contain the arc $uv$, the arc $vu$, both arcs or none of them. The pseudorange configuration is the vector $\myvector{p} = \begin{pmatrix}	\myvector{x}_1^\intercal & \dots & \myvector{x}_n^\intercal & \beta_1 & \dots & \beta_n \end{pmatrix}^\intercal \in \R^{n(d+1)}$ and the pseudorange framework is the pair $(\Gamma, \myvector{p})$. Note that the spatial parameters have been grouped for convenience.
	An arc $uv$ imposes the constraint:
	\begin{equation}\label{eq: Pseudo-range constraint}
		f_{P, uv}(\myvector{p}) \triangleq \norm{\myvector{x}_u - \myvector{x}_v} + \beta_v - \beta_u = \rho_{uv},
	\end{equation}
	where $\rho_{uv}$ is a given pseudorange.
	The definition of congruence, equivalence, and rigidity of pseudorange frameworks are the same as introduced in Section~\ref{ssection: Generalities on rigidity}. However, $K$ is now the complete directed graph whose arc set is $E_K = \set{uv \in V^2, u \ne v}$.
	
	The complexity of pseudorange rigidity comes from the asymmetry of $\Gamma$. Without this asymmetry, the problem would be of no interest. Indeed, two opposite pseudoranges $\rho_{uv}$ and $\rho_{vu}$ constrain both the distance and the bias difference between the agents:
		\begin{align}
			\norm{\myvector{x}_u - \myvector{x}_v} &= \frac{\rho_{uv} + \rho_{vu}}{2}, & \beta_u - \beta_v &= \frac{\rho_{vu} - \rho_{uv}}{2}. \label{eq: Symmetrical constraints}
		\end{align} In that case, the spatial constraints and the bias constraints can be separated. To {rigidify} the biases, the graph must be connected, and to {rigidify} the positions, the graph must be distance rigid. As distance rigid graphs are connected, pseudorange rigidity and distance rigidity are equivalent for symmetrical pseudorange graphs. The study of pseudorange rigidity is however far less trivial when $\Gamma$ is not symmetrical.
	
	The decoupling between the spatial and the bias constraints can also be performed if $\Gamma$ has a spanning tree of symmetrical arcs. In that case, all the bias differences are also set, and pseudorange constraints become distance constraints. Hence, if $\Gamma$ has a spanning tree of symmetrical arcs, pseudorange rigidity and distance rigidity are also equivalent. Our main result in Section~\ref{sec: Characterization of generic pseudorange rigidity} states that in fact, the spanning tree does not need to be formed of symmetrical arcs: a pseudorange graph is rigid if a distance rigid graph can be extracted while the remaining arcs form a connected graph. This result is based on the characterization of pseudorange infinitesimal rigidity, and therefore on the characterization of the rank of the pseudorange rigidity matrix.
	
	Section \ref{ssec: Examples of pseudorange frameworks} proposes four examples of pseudorange frameworks to underline the importance of the orientation of the arcs, and then Section \ref{ssec: Pseudorange generic rigidity} introduces pseudorange generic rigidity. The main results, Theorem~\ref{the: Generic rank} and its Corollaries~\ref{cor: Generic rigidity}-\ref{cor: Decomposition}, are stated in Section~\ref{sec: Characterization of generic pseudorange rigidity}.
	
	\subsection{Examples of pseudorange frameworks}\label{ssec: Examples of pseudorange frameworks}
	
	\begin{figure}
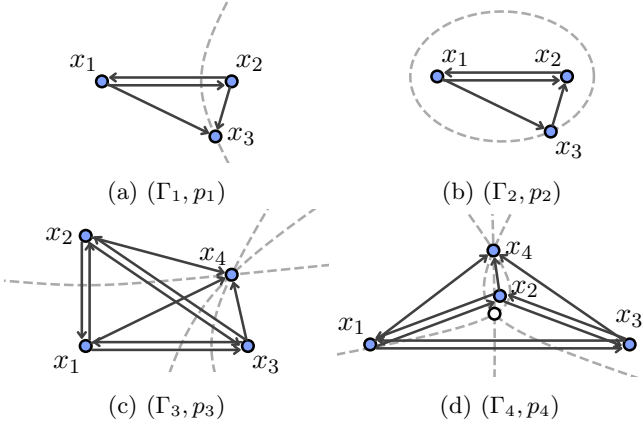

		\centering
		\null\hfill
		\begin{subfigure}{.49\linewidth}
			\centering
			\resizebox{\linewidth}{!}{\input{fig/pseudorange_graph_hyperbolic_motion.pgf}}
			\caption{$(\Gamma_1, p_1)$}
			\label{sfig: Two arcs hyperbolic constraint illustration}
		\end{subfigure}
		\hfill
		\begin{subfigure}{.49\linewidth}
			\centering
			\resizebox{\linewidth}{!}{\input{fig/pseudorange_graph_elliptical_motion.pgf}}
			\caption{$(\Gamma_2, p_2)$}
			\label{sfig: Two arcs elliptical constraint illustration}
		\end{subfigure}
		\hfill
		\begin{subfigure}{.49\linewidth}
			\centering
			\resizebox{\linewidth}{!}{\input{fig/pseudorange_graph_rigid2.pgf}}
			\caption{$(\Gamma_3, p_3)$}
			\label{sfig: Rigid framework}
		\end{subfigure}
		\begin{subfigure}{.49\linewidth}
			\centering
			\resizebox{\linewidth}{!}{\input{fig/pseudorange_graph_rigid.pgf}}
			\caption{$(\Gamma_4, p_4)$}
			\label{sfig: Rigid framework 2}
		\end{subfigure}
		\hfill\null
		\caption{Examples of $2$-dimensional pseudorange frameworks. The positions of the agents are represented by blue circles, the pseudorange constraints by arrows. The dashed curves are construction lines for the position of the last agent and the white circle in Fig.~\ref{sfig: Rigid framework 2} is another possible position for the $4$th agent. The bias axis is not represented.}
		\label{fig: Pseudorange frameworks}
	\end{figure}
	
	Fig.~\ref{fig: Pseudorange frameworks} presents four $2$-dimensional pseudorange frameworks. Each agent has $2$ spatial coordinates plus $1$ clock bias. This paragraph investigates their rigidity.
	
	First, consider the framework in Fig.~\ref{sfig: Two arcs hyperbolic constraint illustration}. This framework has $4$ arcs corresponding to the pseudoranges from $1$ to $2$, from $2$ to $1$, from $1$ to $3$ and from $2$ to $3$. They constrain the following equations:
		\begin{subequations}
			\begin{align}
				\norm{\myvector{x}_1 - \myvector{x}_2} + \beta_2 - \beta_1 &= \rho_{1,2}, \label{eq: Conic example PSR 1 to 2}\\
				\norm{\myvector{x}_1 - \myvector{x}_2} + \beta_1 - \beta_2 &= \rho_{2,1}, \label{eq: Conic example PSR 2 to 1} \\
				\norm{\myvector{x}_1 - \myvector{x}_3} + \beta_3 - \beta_1 &= \rho_{1,3}, \label{eq: Conic example PSR 1 to 3} \\
				\norm{\myvector{x}_2 - \myvector{x}_3} + \beta_3 - \beta_2 &= \rho_{2,3}. \label{eq: Conic example PSR 2 to 3}
			\end{align}
		\end{subequations}
		Since both pseudoranges between agents $1$ and $2$ are constrained, their distance and bias difference are set using \eqref{eq: Symmetrical constraints}. Moreover, since the pseudoranges $\rho_{1,3}$ and $\rho_{2,3}$ are also constrained, subtracting \eqref{eq: Conic example PSR 2 to 3} from \eqref{eq: Conic example PSR 1 to 3} gives:
		\begin{equation}
			\norm{\myvector{x}_1 - \myvector{x}_3} - \norm{\myvector{x}_2 - \myvector{x}_3} = \rho_{1,3} - \rho_{2,3} + \beta_1 - \beta_2 = \text{cst}.
		\end{equation}
		Therefore, the position of agent~$3$ lies on a branch of hyperbola whose foci are $\myvector{x}_1$ and $\myvector{x}_2$. This branch is represented by the dashed line in the figure. The bias $\beta_3$ is obtained by reinjecting the distance into either \eqref{eq: Conic example PSR 1 to 3} or \eqref{eq: Conic example PSR 2 to 3} and is not constant along this curve: it decreases as the distance $\norm{\myvector{x}_1 - \myvector{x}_3}$ increases. Moving agent~$3$ along this $3$-dimensional curve of positions creates a non-congruent but equivalent configuration, therefore, this framework is flexible.
	
	The second framework in Fig.~\ref{sfig: Two arcs elliptical constraint illustration} is very similar. The only difference lies in the direction of the arc between $2$ and $3$. In Fig.~\ref{sfig: Two arcs hyperbolic constraint illustration}, $\rho_{2,3}$ was constrained whereas it is now $\rho_{3,2}$. This transforms \eqref{eq: Conic example PSR 2 to 3} to:
	\begin{align}
		\norm{\myvector{x}_2 - \myvector{x}_3} + \beta_2 - \beta_3 &= \rho_{3,2}. \label{eq: Conic example PSR 3 to 2}
	\end{align}
	Moreover, summing \eqref{eq: Conic example PSR 1 to 3} and \eqref{eq: Conic example PSR 3 to 2} gives:
	\begin{equation}
		\norm{\myvector{x}_1 - \myvector{x}_3} + \norm{\myvector{x}_2 - \myvector{x}_3} = \rho_{1,3} + \rho_{3,2} + \beta_1 - \beta_2 = \text{cst}.
	\end{equation}
	Consequently, $\myvector{x}_3$ lies on an ellipse, also represented by a dashed line in the figure. Similarly, moving agent~$3$ on this curve of positions creates a non-congruent but equivalent configuration and this framework is flexible. These first two examples underline how important the orientations of the arcs are: flipping an arc changes the possible deformations of the framework.
	
	The third and fourth frameworks in Figures~\ref{sfig: Rigid framework} and \ref{sfig: Rigid framework 2} have the same graph which is more complex. The first three agents are fully connected, therefore, all the distances and bias differences between them are constrained: their relative positions and biases are set. The $4$th agent is connected to each of them by one unique arc. Each pair of arcs constrains the position $\myvector{x}_4$ to lie on a branch of hyperbola as in the first example. These curves are also represented by dashed lines.
	For the third framework, they intersect once at $\myvector{x}_4$, this is the only suitable position for the $4$th agent. There are no equivalent but non-congruent frameworks, therefore, $(\Gamma_3, p_3)$ is globally rigid.
	For the fourth framework, they intersect twice: at $\myvector{x}_4$ of course and at a second point represented by a white circle. Those two points are suitable positions for agent~$4$: placing it in one of these loci (with the corresponding bias) satisfies all the constraints. However, agent~$4$ cannot \emph{move} so the framework is rigid but not globally rigid. At these two loci, the associated biases are different since, for example, the distances to $\myvector{x}_2$ are different.
	
	\subsection{Pseudorange generic rigidity}\label{ssec: Pseudorange generic rigidity}
	
	This section introduces pseudorange generic rigidity and its link with pseudorange infinitesimal rigidity. In particular, the properties of the pseudorange rigidity matrix are discussed. 
	\begin{definition}
		A pseudorange configuration $\myvector{p}$ is said to be \emph{generic} if the $nd$ spatial coordinates of the agents are not a root of any non-trivial polynomial with integer coefficients. In this case, the framework $(\Gamma, \myvector{p})$ is also said to be generic.
	\end{definition}
	
	Genericity ensures, for example, that no three agents are aligned or that no four agents are on a plane. In practice, the agents are never perfectly aligned and the configuration is generic. Formally, the set of non-generic configurations is defined as the roots of the polynomials with integer coefficients. Therefore, it is countable and has a zero measure. Consequently, \emph{almost} every configuration is generic. Note that this property considers the positions of the agents but not their clock offsets.
	
	As for distance frameworks, the rigidity of generic pseudorange frameworks is equivalent to their infinitesimal rigidity, as stated in the following lemmas.
	\begin{lemma}\label{lem: Infinitesimal rigidity implies rigidity}
		If a pseudorange framework $(\Gamma, \myvector{p})$ is infinitesimally rigid, then it is rigid.
	\end{lemma}
	\begin{lemma}\label{lem: Equivalence rigidities}
		Let $(\Gamma, \myvector{p})$ be a generic pseudorange framework. $(\Gamma, \myvector{p})$ is rigid if and only if it is infinitesimally rigid.
	\end{lemma}
	The proofs of Lemmas~\ref{lem: Infinitesimal rigidity implies rigidity} and \ref{lem: Equivalence rigidities} are adaptions of proofs on distance rigidity of usual distance frameworks and will not be detailed. They can be found \eg{} in \cite{asimow1978rigidity} or \cite{gluck1975almost}.
	
	Consequently, to study the generic rigidity of pseudorange frameworks, we focus on their rigidity matrices.
	To match the usual form of rigidity matrix, we define the pseudorange rigidity matrix of the framework as follows:
	\begin{equation}\label{eq: Pseudorange rigidity matrix}
		\mymatrix{R}_P(\Gamma, \myvector{p}) \triangleq \mymatrix{D}(\Gamma,\myvector{p})\frac{\partial \myvector{F}_P(\Gamma, \myvector{p})}{\partial \myvector{p}} \in \R^{m\times n(d+1)},
	\end{equation}
	 where $\mymatrix{D}(\Gamma,\myvector{p}) = \diag\left(\set{\norm{\myvector{x}_u-\myvector{x}_v}, uv \in E}\right)$ is the diagonal matrix whose $i$th entry is the distance between the points connected by the $i$th arc. With this definition, the pseudorange rigidity matrix has the following structure:
	\begin{equation}\label{eq: Decomposition M}
		\mymatrix{R}_P(\Gamma, \myvector{p}) = \begin{bmatrix}
			\mymatrix{R}_D(\Gamma, \myvector{p}) & \mymatrix{R}_S(\Gamma, \myvector{p})
		\end{bmatrix},
	\end{equation}
	where $\mymatrix{R}_D(\Gamma,\myvector{p}) \in \R^{m \times nd}$ is the distance rigidity matrix of the framework (where $\Gamma$ is viewed as an undirected multigraph), and $\mymatrix{R}_S(\Gamma, \myvector{p}) \in \R^{m \times n}$ is a rigidity matrix associated with the synchronizations of the clocks. It corresponds to the clock offset variables and is defined as:
	\begin{equation}\label{eq: Decomposition B}
		\mymatrix{R}_S(\Gamma,\myvector{p}) = \mymatrix{D}(\Gamma,\myvector{p})\mymatrix{B}(\Gamma)^\intercal,
	\end{equation}
	where $\mymatrix{B}(\Gamma)$ denotes of the incidence matrix of the graph $\Gamma$, see \eg{} \cite[p.~54]{bollobas1998modern}. For example, the rigidity matrix of the pseudorange framework in Fig.~\ref{sfig: Two arcs hyperbolic constraint illustration} is given in \eqref{eq: Example rigidity matrix} at the top of the next page.
	\begin{figure*}
		\normalsize
		\begin{equation}\label{eq: Example rigidity matrix}
			R_{P}(\Gamma,\myvector{p}) = \left[\begin{array}{ccc|ccc}
				\myvector{x}^\intercal_1 - \myvector{x}^\intercal_2 & \myvector{x}^\intercal_2 - \myvector{x}^\intercal_1 & \myvector{0}^\intercal & -\norm{\myvector{x}_1 - \myvector{x}_2} & \norm{\myvector{x}_1 - \myvector{x}_2} & 0\\
				\myvector{x}^\intercal_1 - \myvector{x}^\intercal_2 & \myvector{x}^\intercal_2 - \myvector{x}^\intercal_1 & \myvector{0}^\intercal & \norm{\myvector{x}_1 - \myvector{x}_2} & -\norm{\myvector{x}_1 - \myvector{x}_2} & 0\\
				\myvector{x}^\intercal_1 - \myvector{x}^\intercal_3 & \myvector{0}^\intercal & \myvector{x}^\intercal_3 - \myvector{x}^\intercal_1 & -\norm{\myvector{x}_1 - \myvector{x}_3} & 0 & \norm{\myvector{x}_1 - \myvector{x}_3}\\
				\myvector{0}^\intercal & \myvector{x}^\intercal_2 - \myvector{x}^\intercal_3 & \myvector{x}^\intercal_3 - \myvector{x}^\intercal_2 & 0 & -\norm{\myvector{x}_2 - \myvector{x}_3} & \norm{\myvector{x}_2 - \myvector{x}_3}
			\end{array}\right]
		\end{equation}
	\end{figure*}
	
	From the decomposition \eqref{eq: Decomposition M}, the rank of the pseudorange rigidity matrix is lower than the sum of the ranks of each block. The rank of the distance rigidity matrix is bounded by a quantity $S_D(n,d)$ \cite{asimow1978rigidity} defined as:
	\begin{equation}
		S_D(n,d) \triangleq \left\{\begin{array}{ll}
			nd - \binom{d+1}{2} & \text{if } n \ge d + 1,\\
			\binom{n}{2} & \text{if } n \le d.
		\end{array}\right..
	\end{equation}
	Moreover, the maximal rank of an incidence matrix between $n$ agents is $n-1$: the vector filled with ones is always in the cokernel of the incidence matrix.
	As a result, the rank of the rigidity matrix is bounded by a quantity $S_P(n,d)$ that depends on both the number of agents $n$ and the dimension $d$:
	\begin{equation}\label{eq: Bound rank}
		\rank \mymatrix{R}_P(\Gamma,\myvector{p}) \le S_P(n,d) \triangleq S_D(n,d) + n - 1.
	\end{equation}
	The interpretation of \eqref{eq: Bound rank} is that the trivial velocity vectors are composed of the $d$ spatial translations, the $d(d-1)/2$ spatial rotations and the bias translation.
	
	\begin{definition}\label{def: Generic configuration}
		A pseudorange framework $(\Gamma, \myvector{p})$ is said to be \emph{infinitesimally rigid} if $\rank \mymatrix{R}_P(\Gamma,\myvector{p}) = S_P(n,d)$.
	\end{definition}
	
	The next section proves that, generically, infinitesimal rigidity is a property of the graph.
	
	\section{Characterization of generic pseudorange rigidity}\label{sec: Characterization of generic pseudorange rigidity}
	
	In this section, we provide a complete characterization of the rigidity of generic pseudorange frameworks in terms of distance rigidity.
	We prove that pseudorange rigidity is a generic property of the underlying undirected graphs of the frameworks. We explicit the rank of the rigidity matrix according to possible \emph{decompositions} of the graph. Our approach is similar to the one used for distance frameworks in \eg{} \cite{hendrickson1992conditions}. The rank is considered as the order of a highest order non-vanishing minor of the matrix.
	
	For distance frameworks the minors of the rigidity matrix $\mymatrix{R}_D(G, \myvector{p})$ are polynomials with integer coefficients in the $nd$ coordinates of the agents. By definition of a generic configuration, a minor vanishes for a generic configuration only if it is the null function (in this case, it vanishes for every configuration).
	For pseudorange rigidity matrices, the minors are also functions of the $nd$ coordinates of the agents. However, contrary to the minors of distance rigidity matrices, these minor functions are not polynomials. They belong to a larger space of functions we call $\L$. Recall that \eqref{eq: Example rigidity matrix} gives an example of a pseudorange rigidity matrix.
	
	We define the space $\L$ as follows. For any set of edges $E \subseteq \{uv \mid 1 \le u < v \le n\}$, we define the space $\L(E)$ as:
	\begin{equation}\label{eq: Definition of L}
		\L(E) = \set{ \sum_{F \in \rP(E)} P_F \prod_{uv \in F} D_{u,v} \mid P_F \in \K }
	\end{equation}
	where $\K = \Q\left(X_1^{(1)}, \dots, X_n^{(d)}\right)$ is the field of rational functions with integer coefficients in $nd$ variables, $\rP(E)$ denotes the power set of $E$, and $D_{u,v}$ is the distance function in $nd$ variables:
	\begin{equation}
		D_{u,v} : (x_1^{(1)}, \dots, x_n^{(d)})  \mapsto  \sqrt{\sum_{i=1}^{d} \left(x_u^{(i)} - x_v^{(i)}\right)^2 }.
	\end{equation}
	Then, the minors of the pseudorange rigidity matrix of a framework $(\Gamma, \myvector{p})$ whose underlying undirected graph is $\tilde \Gamma = (V, \tilde E)$ belong to the space $\L = \L(\tilde E)$.
	
	The definitions of $\L(E)$ and $\K$ may seem overly complex. They have both been chosen to provide a field structure to $\L(E)$ as stated in the following lemma.
	\begin{lemma}\label{lem: Field extension}
		Let $E \subseteq \{uv \mid 1 \le u < v \le n\}$ be a set of edges and $m = \abs{E}$. Then, $\L(E) / \K $ is a field extension of degree $2^m$. Furthermore, the family $\set{\prod_{uv \in F}D_{u,v} \mid F \in \rP(E)}$ is a basis of $\L(E)$ viewed as a $\K$-vector space. We call this basis the \emph{natural basis} of $\L(E)$. 
	\end{lemma}
	
	Lemma~\ref{lem: Field extension} involves several elements from field theory. They will not be discussed here, nonetheless, these concepts can be found \eg{} in \cite{roman2005field} and the proof of Lemma~\ref{lem: Field extension} is provided in Appendix~\ref{ap: Proof of algebraic lemmas} for the sake of completeness.
	Solely the implications of this lemma are explained here. The first important point is that $\L$ is a field. Therefore, every nonzero element has a multiplicative inverse. Second, $\L$ is a $\K$-vector space of dimension $2^m$ and one natural basis is known. For example, if $E = \set{ab, bc, ac}$, the natural basis of $\L$ has $8$ elements: the constant function equals to $1$, the three distance functions $D_{a,b}$, $D_{b,c}$ and $D_{a,c}$, the three products of two distance functions $D_{a,b}D_{b,c}$, $D_{b,c}D_{a,c}$ and $D_{a,c}D_{a,b}$, and the product of the three distance functions $D_{a,b}D_{b,c}D_{a,c}$.
	The last important consequence of Lemma~\ref{lem: Field extension} is that the polynomials $P_F$ involved in \eqref{eq: Definition of L} are unique as they are the coordinates on the natural basis.
	
	By definition, the only polynomial with integer coefficents that vanishes at a generic point of $\R^{nd}$ is the zero polynomial. The structure of $\L$ allows to extend this property to the field $\L$ as explained in the following lemma.
	\begin{lemma}\label{lem: Extension generic nullity}
		Let $E \subseteq \{uv \mid 1 \le u < v \le n\}$ be a set of edges and $f \in \L(E)$.
		If $\exists \myvector{x} \in \R^{nd} $ a generic vector such that $f(\myvector{x}) = 0$, then $f = 0$.
	\end{lemma}
	\begin{proof}
		Let $\myvector{x}$ be a generic vector and $m = \abs{E}$. For every $E' \subseteq E$, by Lemma~\ref{lem: Field extension}, $\L(E')$ is a field and a $\K$-vector space of dimension $2^{\abs{E'}}$ whose natural basis is composed of the products between the distance functions.
		
		Let us prove the lemma by induction on the number of distance functions appearing in the expression of $f$. Let us prove that:
		$\forall k \in \set{0, \dots, m}$, if $f \in \L(E')$ with $\abs{E'} = k$ and if $f(\myvector{x}) = 0$, then $f = 0$.
		
		\emph{Base case:} If $k = 0$, $f$ is a rational function with integer coefficients, \ie{} $f = P/Q$ with $P$ and $Q$ two polynomials with integer coefficients. By definition, since $\myvector{x}$ is generic and $P(\myvector{x}) = 0$, $P$ is the null function and therefore $f = 0$.
		
		\emph{Inductive step:} Let $E'$ have cardinality $k+1$ with $k \ge 0$, $f \in \L(E')$ with $f(\myvector{x}) = 0$, and $uv \in E'$. Any function $h \in \L(E')$ can be uniquely decomposed, by separating the natural basis of $\L(E')$, as $h_1 + D_{u,v} h_2$ with $h_1, h_2 \in \L(E'\setminus{\set{uv}})$.
		Let $f = f_1 + D_{u,v} f_2$ be this decomposition applied to $f$. Furthermore, let $\bar f = f_1 - D_{u,v} f_2$ and $g = f\bar f = f_1^2 - D_{u,v}^2 f_2^2$. As $D_{u,v}^2$ is a polynomial, $g \in \L(E'\setminus{\set{uv}})$ where $E'\setminus{\set{uv}}$ has cardinality $k$. Since $f$ vanishes at $\myvector{x}$, $g$ also vanishes at $\myvector{x}$ and by the induction hypothesis, $g = 0$.
		Therefore, since $\L(E')$ is a field, either $f = 0$ or $\bar f = 0$. By definition, $f$ and $\bar f$ have the same coordinates up to a sign in the natural basis, thus $f = 0$.
	\end{proof}
	
	As the $nd$ spatial coordinates of a generic framework form a generic point of $\R^{nd}$, Lemma~\ref{lem: Extension generic nullity} implies that the rank of the pseudorange rigidity matrix is a generic property of its graph $\Gamma$ and of its underlying undirected graph $\tilde\Gamma$. We prove in the sequel a stronger result: the rank can be expressed using \emph{decompositions} of $\tilde \Gamma$.
	\begin{definition}
			Let $\tilde \Gamma = (V, \tilde E)$ be the underlying undirected graph of a directed graph $\Gamma = (V, E)$. Denote $E_1$ the set of edges that appear once in $\tilde E$ and $E_2$ the set of edges that appear twice in $\tilde E$.
			Two simple graphs $G_D = (V, E_D)$ and $G_S = (V, E_S)$ are said to form a \emph{decomposition} of the multigraph $\tilde \Gamma = (V, \tilde E)$ if:
			\begin{enumerate}
				\item $E_D \cup E_S = E_1 \cup E_2$.
				\item $E_D \cap E_S = E_2$.
			\end{enumerate}
			A decomposition $(G_D, G_S)$ of $\tilde \Gamma$ is denoted as $\tilde \Gamma = G_D\cup G_S$.
	\end{definition}
	In other words, a decomposition of an undirected multigraph is a splitting of its edges into two simple graphs. Of course, an undirected multigraph often admits more than one decomposition. The subscript \guillemets{$D$} has been chosen as $G_D$ will be searched as a distance rigid graph. The subscript \guillemets{$S$} has been chosen as $G_S$ will be searched as a connected graph in order to synchronize the clocks.
	
	The rank of a distance rigidity matrix is a generic property of its graph \cite{asimow1978rigidity}, let us denote this generic rank simply as $\rank \mymatrix{R}_D(G)$. Furthermore from \eqref{eq: Decomposition B}, the rank of $\mymatrix{R}_S(G, \myvector{p})$ is also a generic property of the graph, it is equal to the rank of the incidence matrix of $G$. Let us denote this generic rank similarly as $\rank \mymatrix{R}_S(G)$.
	With these notations, we are now in a position to state our main result.
	
	\begin{theorem}\label{the: Generic rank}
		Let $(\Gamma, \myvector{p})$ be a generic pseudorange framework whose underlying undirected multigraph is $\tilde \Gamma$ and denote $r = \rank \mymatrix{R}_P(\Gamma, \myvector{p})$. Then:
		\begin{equation}
			r = \max_{\tilde \Gamma = G_D \cup G_S} \rank \mymatrix{R}_D(G_D) + \rank \mymatrix{R}_S(G_S)
		\end{equation}
		where the maximum is taken over the decompositions of $\tilde\Gamma$.
	\end{theorem}
	\begin{proof}
		First, consider a decomposition $(G_D, G_S)$ of $\tilde \Gamma$ and let us prove that $r \ge \rank \mymatrix{R}_D(G_D) + \rank \mymatrix{R}_S(G_S)$. Denote $r_D = \rank \mymatrix{R}_D(G_D)$ and $r_S = \rank \mymatrix{R}_S(G_S)$. There exist two square submatrices of $\mymatrix{R}_D(G_D, \myvector{p})$ and $\mymatrix{R}_S(G_S, \myvector{p})$ of size $r_D$ and $r_S$ with non-null determinants. Let $F_D$ and $F_S$ denote the edge sets associated with their rows, and $C_D$ and $C_S$ the indices of their columns. Then consider $\mymatrix{N}(\Gamma ,\myvector{p})$ the square submatrix of $\mymatrix{R}_P(\Gamma,\myvector{p})$ of size $r_D + r_S$ associated with the columns $C_D \cup C_S$ and the rows $E = F_D \cup F_S$. As in \eqref{eq: Decomposition M}, write $\mymatrix{N}(\Gamma,\myvector{p}) = \begin{bmatrix}
			\mymatrix{N}_D(\Gamma,\myvector{p}) & \mymatrix{N}_S(\Gamma,\myvector{p})
		\end{bmatrix}$.
		By employing the Laplace expansion theorem, see \eg{} \cite[Section~0.8.9]{horn2012matrix}, on the determinant of $\mymatrix{N}(\Gamma, \myvector{p})$:
		\begin{equation}\label{eq: Decomposition N}
			\det \mymatrix{N}(\Gamma) = \sum_{F\subseteq E: \abs{F} = r_S} \pm \det \mymatrix{N}_D(E\setminus F) \det \mymatrix{N}_S(F),
		\end{equation}
		where the sum is over the subsets of $E$ of cardinality $r_S$, $\det \mymatrix{N}(\Gamma)$ denotes the function $p\mapsto \det \mymatrix{N}(\Gamma,\myvector{p})$, $\mymatrix{N}_X(F)$ denotes the submatrix of $\mymatrix{N}_X(\Gamma)$ induced by the rows of $F$, and the \guillemets{$\pm$} depend on the signs of the permutations in the expansion. 
		From \eqref{eq: Decomposition B}, the rows of $\mymatrix{R}_S(F,p)$ are proportional to the columns of the incidence matrix induced by $F$. Therefore if there is a cycle composed of edges in $F$ and vertices in $U$, where $U$ denotes the set of vertices associated with the columns $C_S$, then $\det \mymatrix{N}_S(F)$ is the null function. Furthermore, if there is no such cycle, $\det \mymatrix{N}_S(F) = \pm \prod_{uv \in F} D_{u,v}$. Thus, by denoting $\rF$ the set of $F \subseteq E$ having cardinality $r_S$ without cycle in $U$:
		\begin{equation}\label{eq: Decomposition N 2}
			\det \mymatrix{N}(\Gamma) = \sum_{F\in \rF} P_F
			\prod_{uv \in F} D_{u,v} \in \L(E).
		\end{equation}
		where $P_F = \pm \det \mymatrix{N}_D(E\setminus F) \in \K$. Equation \eqref{eq: Decomposition N 2} is the decomposition of $\det \mymatrix{N}(\Gamma)$ in the natural basis of $\L(E)$. Since $\det \mymatrix{N}_D(E\setminus F_S, \myvector{p}) = \det \mymatrix{N}_D( F_D, \myvector{p}) \ne 0$, the coefficient associated with $F_S$ is not null and $\det \mymatrix{N}(\Gamma)\neq 0$. Therefore, as $\myvector{p}$ is generic according to Lemma~\ref{lem: Extension generic nullity}, $\det \mymatrix{N}(\Gamma,\myvector{p})\neq 0$ and $r \ge r_D+ r_S$. Thus:
		\begin{equation*}
			r \ge \max_{G_D \cup G_S = \tilde \Gamma} \rank \mymatrix{R}_D(G_D) + \rank \mymatrix{R}_S(G_S).
		\end{equation*}
		
		Conversely, consider a submatrix $\mymatrix{N}(\Gamma, \myvector{p})$ of $\mymatrix{R}_P(\Gamma, \myvector{p})$ of size $r$ with a non-null determinant. Denote as $F$ its rows. Write similarly $\mymatrix{N}(\Gamma,\myvector{p}) = \begin{bmatrix}
			\mymatrix{N}_D(\Gamma,\myvector{p}) & \mymatrix{N}_S(\Gamma,\myvector{p})
		\end{bmatrix}$. With the same notations as in \eqref{eq: Decomposition N}, since $\det \mymatrix{N}(\Gamma) \neq 0$, there exists $F_S \subseteq F$ such that $\mymatrix{N}_D(F\setminus F_S) \neq 0$ and $\det \mymatrix{N}_S(F_S) \neq 0$. Let $F_D = F \setminus F_S$, $r_D = \abs{F_D}$ and $r_S = \abs{F_S}$. Finally, set $E_S = F_S \cup E_2$ and $E_D = E_2 \cup (E_1 \setminus F_S)$ where $E_1$ and $E_2$ denote the edge sets of the single and double edges of $\tilde \Gamma$. We can verify that $G_D = (V, E_D)$ and $G_S = (V, E_S)$ form a decomposition of $\tilde\Gamma$. Furthermore, by construction, $\rank \mymatrix{R}_D(G_D) \ge r_D$ and $\rank \mymatrix{R}_S(G_S) \ge r_S$. Thus:
		\begin{equation*}
			\rank \mymatrix{R}_D(G_D) + \rank \mymatrix{R}_S(G_S) \ge r_D + r_S = r,
		\end{equation*}
		concluding the proof of Theorem~\ref{the: Generic rank}.
		
	\end{proof}
	Theorem~\ref{the: Generic rank} implies that two generic pseudorange frameworks having the same underlying undirected graph have rigidity matrices of the same rank. Consequently, from the definition of infinitesimal rigidity, Theorem~\ref{the: Generic rank} has the following corollary.
	\begin{corollary}\label{cor: Generic rigidity}
		Let $\tilde \Gamma$ be an undirected pseudorange graph. Either every generic $d$-dimensional pseudorange framework whose underlying undirected pseudorange graph is $\tilde \Gamma$ is rigid or none of them is. In this former case, $\tilde \Gamma$ is said to be rigid in $\R^d$.
	\end{corollary}
	This corollary is illustrated by the pseudorange frameworks of Fig.~\ref{fig: Pseudorange frameworks}. The frameworks $(\Gamma_1, p_1)$ and $(\Gamma_2, p_2)$ have the same underlying undirected graph, while their graphs are both flexible. Note that however the admissible deformations are different. Similarly, the frameworks $(\Gamma_3, p_3)$ and $(\Gamma_4, p_4)$ have the same graph and are both rigid.
	
	The second main consequence of Theorem~\ref{the: Generic rank} is the characterization of the rigidity of the underlying undirected graph.
	\begin{corollary}\label{cor: Decomposition}
		Let $\tilde \Gamma$ be an undirected pseudorange graph. $\tilde \Gamma$ is rigid in $\R^d$ if and only if there exists a decomposition $(G_D, G_S)$ of $\tilde \Gamma$ such that $G_D$ is distance rigid in $\R^d$ and $G_S$ is connected.
	\end{corollary}
	\begin{proof}
		Let $(\Gamma, \myvector{p})$ be a generic pseudorange framework having $\tilde\Gamma$ for underlying undirected pseudorange graph. $\tilde\Gamma$ is rigid in $\R^d$ if and only if $\rank \mymatrix{R}_P(\Gamma, \myvector{p}) = S_P(n,d)$. For any graph $G_D$, $\rank \mymatrix{R}_D(G_D) \le S_D(n,d)$ with equality if and only if $G_D$ is distance rigid in $\R^d$. Similarly, for any graph $G_S$, $\rank \mymatrix{R}_S(G_S) = \rank \mymatrix{B}(G_S) \le n -1$ with equality if and only if $G_S$ is connected.
		As by definition $S_P(n,d) = S_D(n,d) + n -1$, according to Theorem~\ref{the: Generic rank}, $\rank \mymatrix{R}_P(\Gamma, \myvector{p}) = S_P(n,d)$ if and only if there exists a decomposition $(G_D, G_S)$ that achieves both equalities.
	\end{proof}
	
	
	Corollary~\ref{cor: Decomposition} gives an interpretation to the rigidity of pseudorange frameworks. To be rigid a pseudorange graph should have a decomposition into a distance rigid graph and a connected graph. The distance rigid graph sets the positions of the agents while the connected graph synchronizes their clocks.
	This decomposition may be viewed as a decoupling of the space and clock variables.
	From a combinatorial point of view, Theorem~\ref{the: Generic rank} can be stated using \emph{matroid} theory, see \eg{} \cite{oxley2006matroid}. The pseudorange rigidity matroid is defined on the edges of the directed graph $\Gamma$. Its independent sets are the sets of edges that generate independent rows in generic rigidity matrices. Theorem~\ref{the: Generic rank} states that the pseudorange rigidity matroid is the union of the distance rigidity matroid with the cycle matroid.

	Testing pseudorange rigidity requires testing distance rigidity. In 2D, there exist efficient algorithms, \ie{} that run in polynomial time, to verify the rigidity of graphs, see \eg{} \cite{jacobs1997algorithm}. However, the distance rigidity matroid is not characterized in 3D and no deterministic algorithm can be employed. In \cite{gortler2010characterizing}, the authors proposed an efficient alternative by employing randomized rigidity tests. The idea is to infer infinitesimal rigidity by computing the rank of a randomly generated configuration with (large) integer coordinates. This test will never return a false positive and the probability to return a false negative is bounded. By repeating the test several times, the probability of error can be reduced to an acceptable level.
	
	An extension of this work could focus on the global rigidity of pseudorange frameworks. As distance rigidity, global distance rigidity is a generic property of the graphs \cite{connelly2005generic, gortler2010characterizing}. Global pseudorange rigidity is a generic property neither of the underlying undirected graph nor of the directed graph. Indeed, the two rigid frameworks in Figures~\ref{sfig: Rigid framework} and \ref{sfig: Rigid framework 2} form a counterexample: they have the same graph but one is globally rigid and the other is not. This example is an adaptation of the single-receiver single-constellation problem with $4$ satellites and was highlighted in \cite{abel1991existence} (and \cite{schmidt1972new} for the LORAN system). Future work will study global pseudorange rigidity. In particular, we can adapt Corollary~\ref{cor: Decomposition} to make the following conjecture.
	\begin{conjecture}
		Let $(\Gamma, p)$ be a generic pseudorange framework whose underlying undirected pseudorange graph is $\tilde \Gamma$. If there exists a decomposition $(G_D, G_S)$ of $\tilde \Gamma$ such that $G_D$ is globally rigid in $\R^d$ and $G_S$ is connected, then $(\Gamma, p)$ is globally rigid.
	\end{conjecture}
	This conjecture is motivated by the algebraic methods proposed to solve the usual GNSS positioning problem \cite{bancroft1985algebraic, krause1987direct}. When employed with $4$ satellites, they provide two candidate solutions. One of these candidates may imply a nonsense, \eg{} a negative distance, and should be ruled out. We conjecture that the uniqueness of theses weak solutions is a generic property of the graph and that if we weaken the pseudorange constraint to:
	\begin{align}
		\rho_{u,v} &= \pm\norm{\myvector{x}'_u - \myvector{x}'_v} + \beta_{v}' - \beta_{u}',
	\end{align}
	then global rigidity becomes a generic property of the underlying undirected graph. 
	
	\section{GNSS rigidity for cooperative positioning}\label{sec: GNSS rigidity}
	
	\subsection{GNSS rigidity}
	
	This section presents the adaptation of pseudorange rigidity for cooperative GNSS.

	Consider a group of $S$ satellites belonging to $C$ different GNSS constellations and $R$ cooperative receivers. The receivers cooperate by measuring distances. Furthermore, the positions of the satellites are known, and the satellites belonging to the same constellation are synchronized. We focus on the rigidity of the framework formed by the receivers and the satellites. Each agent (satellite or receiver) is parameterized by $\myvector{p}_i = \begin{pmatrix} \myvector{x}_i^\intercal & \beta_i \end{pmatrix}^\intercal \in \R^{d+1}$. This framework has $3$ types of constraints:
		\begin{enumerate}
			\item Pseudorange constraints: from satellites to receivers. They are represented by a directed graph $\Gamma = (V, E_P)$.
			\item Distance constraints: between receivers and between satellites. The distance constraints between receivers are due to the distance measurements, while the distance constraints between satellites are due to the fact that their positions are known. As the positions of the satellites are known, so are their inter-distances. Distance constraints are represented by an undirected graph $G_D = (V, E_D)$.
			\item Synchronization constraints: between satellites. The satellites within a GNSS constellation are synchronized, therefore if two satellites $u$ and $v$ belong to the same GNSS constellation, then $\beta_u = \beta_v$. These constraints are represented in an undirected graph $G_S = (V, E_S)$.
		\end{enumerate}
		These three graphs of constraints are grouped into one graph that we call a GNSS graph $\rG = (\Gamma, G_D, G_S)$. Similarly to pseudorange graphs, we denote as $\tilde \rG = (\tilde\Gamma, G_D, G_S)$ the underlying undirected GNSS graph.
		Fig.~\ref{fig: Example of GNSS graphs} presents the two GNSS graphs associated with the simple cooperative networks introduced in Fig.~\ref{fig: Two receivers examples}.
	\begin{figure}
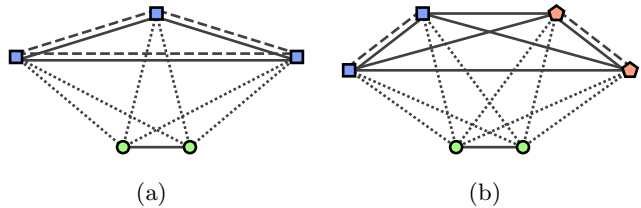

		\begin{subfigure}{.49\linewidth}
			\centering
			{\input{fig/bi_receivers_mono_constellation_grounded_graph.pgf}}
			\caption{}
			\label{sfig: Two receivers mono constellation}
		\end{subfigure}
		\begin{subfigure}{.49\linewidth}
			\centering
			{\input{fig/bi_receivers_bi_constellation_grounded_graph.pgf}}
			\caption{}
			\label{sfig: Two receivers bi constellation}
		\end{subfigure}
		\caption{GNSS graphs associated with the graphs of measurements of Fig.~\ref{fig: Two receivers examples}. The edge sets are represented by: dotted lines for $E_P$, solid lines for $E_D$, and dashed lines for $E_S$.}
		\label{fig: Example of GNSS graphs}
	\end{figure}
	We define a GNSS framework as $(\rG, \myvector{p})$, the combination of the GNSS graph and the pseudorange configuration of the agents. The rigidity matrix of a GNSS framework is:
		\begin{equation}\label{eq: Decomposition GNSS Rigidity matrix}
			\mymatrix{R}_{G}(\rG, \myvector{p}) = \begin{bmatrix}
				\mymatrix{R}_D(\Gamma, \myvector{p}) & \mymatrix{R}_S(\Gamma, \myvector{p}) \\
				\mymatrix{R}_D(G_D, \myvector{p}) & \mymatrix{0} \\
				\mymatrix{0} & \mymatrix{R}_S(G_S, \myvector{p})
			\end{bmatrix},
	\end{equation}
	where the edges of $G_S$ have been oriented.
	
	The theorems for pseudorange frameworks naturally adapt to GNSS frameworks by adapting the decomposition.
	\begin{definition}
		Consider $\tilde \rG = (\tilde \Gamma, G_D, G_S)$ an underlying undirected GNSS graph.
		Two simple graphs $G'_D = (V, E'_D)$ and $G'_S = (V, E'_S)$ are said to form a \emph{decomposition} of $\tilde \rG$ if there exists a decomposition $G_{PD} =(V, E_{PD})$, $G_{PS} = (V, E_{PS})$ of $\tilde \Gamma$ such that:
		\begin{enumerate}
			\item $E'_D = E_D \cup E_{PD}$.
			\item $E'_S = E_S \cup E_{PS}$.
		\end{enumerate}
	\end{definition}
	In other words, the pseudoranges of $\Gamma$ are divided between the graph of distance constraints and the graph of synchronizations.
	Adapting Theorem~\ref{the: Generic rank}, the rigidity of a GNSS framework is generic property of its underlying undirected graph and Corollary~\ref{cor: Decomposition} becomes:
	\begin{theorem}\label{the: GNSS Decomposition}
		Let $\tilde \rG$ be an undirected GNSS graph. $\tilde \rG$ is rigid in $\R^d$ if and only if there exists a decomposition $(G_D', G_S')$ of $\tilde \rG$ such that $G_D'$ is distance rigid in $\R^d$ and $G_S'$ is connected.
	\end{theorem}
	\begin{proof}
		As for pseudorange frameworks, the maximal rank of a GNSS rigidity matrix is $S_P(n,d)$. A distance edge $uv \in G_D$ can only increase the rank of the distance part $\mymatrix{R}_D$. A synchronization edge $uv \in G_S$ can only increase the rank of the bias part $\mymatrix{R}_S$. Similar to the proof of Corollary~\ref{cor: Decomposition}, the rigidity matrix has maximal rank if and only if there exists a decomposition $(G_D', G_S')$ with $G_D'$ and $G_S'$ both having maximal rank.
	\end{proof}	
		
\subsection{Solvability of GNSS problems}
	
Theorem~\ref{the: GNSS Decomposition} provides a new interpretation for the solvability of GNSS problems. In a GNSS problem, the objective is to find the positions and the clock offsets of the receivers (with respect to some constellation time taken as a reference). A problem is said to be solvable if the measurements allow to isolate solutions, \ie{} if the solution set is discrete. There may be several solutions, but having a discrete set is sufficient: the navigation algorithm will track one of them, \eg{} the closest to the previous estimation. Under the natural assumptions that $(i)$ the agents are in a generic configuration, and $(ii)$ the number of satellites is $S \ge d$, the solvability of a GNSS problem is equivalent to the rigidity of the associated GNSS graph. If the GNSS framework is flexible, the measurements are insufficient to identify a specific configuration because of the admissible deformations. On the contrary, if the framework is rigid, there is locally only one configuration realizing the measurements, see \cite{aspnes2006theory} for a detailed analysis on the connections between solvability and rigidity (for distance rigidity).
	
	\begin{figure}
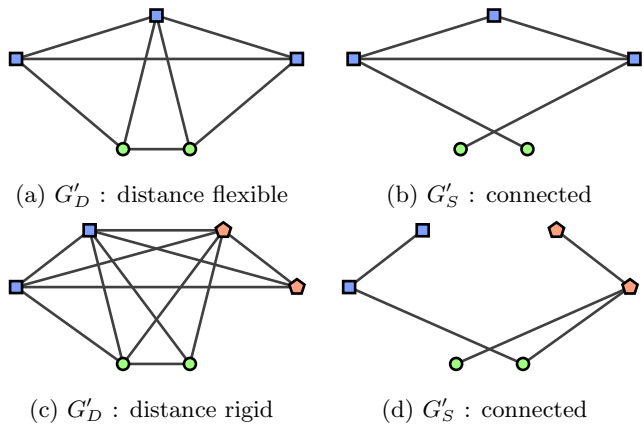

	\begin{subfigure}{.49\linewidth}
		\centering
		{\input{fig/bi_receivers_mono_constellation_rigid_graph.pgf}}
		\caption{$G_D'$ : distance flexible}
		\label{sfig: Two receivers mono constellation rigid graph}
	\end{subfigure}
	\begin{subfigure}{.49\linewidth}
		\centering
		{\input{fig/bi_receivers_mono_constellation_connected_graph.pgf}}
		\caption{$G_S'$ : connected}
		\label{sfig: Two receivers mono constellation connected graph}
	\end{subfigure}
	\begin{subfigure}{.49\linewidth}
		\centering
		{\input{fig/bi_receivers_bi_constellation_rigid_graph.pgf}}
		\caption{$G_D'$ : distance rigid}
		\label{sfig: Two receiver rigid graph}
	\end{subfigure}
	\begin{subfigure}{.49\linewidth}
		\centering
		{\input{fig/bi_receivers_bi_constellation_connected_graph.pgf}}
		\caption{$G_S'$ : connected}
		\label{sfig: Two receiver connected graph}
	\end{subfigure}
	\caption{Decompositions of the GNSS graphs of Fig.~\ref{fig: Example of GNSS graphs}.
		Figs.~\ref{sfig: Two receivers mono constellation rigid graph} and \ref{sfig: Two receivers mono constellation connected graph}: Non-rigid decomposition of the GNSS graph of \ref{sfig: Two receivers mono constellation} ($G_D'$ is not distance rigid in $\R^3$). Figs.~\ref{sfig: Two receiver rigid graph} and \ref{sfig: Two receiver connected graph}: rigid decomposition of the GNSS graph of Fig.~\ref{sfig: Two receivers bi constellation}.}
	\label{fig: Decompositions GNSS graphs}
\end{figure}

GNSS rigidity provides a new understanding on the minimal number of measures required to locate a GNSS receiver. If a receiver is measuring signals from satellites belonging to $C$ different GNSS constellations, $C+d$ satellites are required to locate it. In this case, $C$ pseudorange measurements are used to connect $G_S'$, \ie{} to synchronize the agents, while the other $d$ rigidify the position of the receiver in $G_D'$. Consider now the networks illustrated in Fig.~\ref{sfig: Two receivers mono constellation} and \ref{sfig: Two receivers bi constellation}. In both figures, each receiver receives signals from only $2+C$ satellites. Without cooperation, they cannot be located (in 3D) as it would require that they receive signals from at least $C+3$ satellites. If $C=1$, the cooperation does not allow the agents to be located. Fig.~\ref{fig: Decompositions GNSS graphs} proposes a decomposition of the GNSS graph with the graph $G_S'$ connected. The resulting graph $G_D'$ is not distance rigid in $\R^3$: the agents are connected to only two satellites and can \guillemets{swing} around them. In the bi-constellation scenario of Fig.~\ref{sfig: Two receivers bi constellation}, the problem of the localization of the agents becomes solvable: Fig.~\ref{fig: Decompositions GNSS graphs} presents a decomposition satisfying the condition of Theorem~\ref{the: GNSS Decomposition}.	
Theorem~\ref{the: GNSS Decomposition} also gives the minimal number of measurements required to locate a network.

\begin{lemma}\label{lem: Minimum number}
	The minimum number of pseudorange or distance measurements to locate a network of $R$ receivers using satellites from $C$ different constellations is $R(d+1) + C - 1$.
\end{lemma}
\begin{proof}
	To be localizable, the GNSS graph formed by the receivers and the satellites must be rigid. Therefore, it must have at least $S_P(R+S, d)$ constraints.
	The proof can be performed by induction. $d+1$ constraints are used to rigidify each receiver and $C$ to synchronize the constellations.
\end{proof}

Lemma~\ref{lem: Minimum number} extends the statement given in Section~\ref{sec: Introduction} (positioning $R = 1$ receiver requires at least $3+C$ measurements). Consider the network in Fig.~\ref{fig: Two receivers examples} as an  example for $R \geq 1$.  With $R=2$ receivers in $\R^3$, at least $7 + C$ measurements are needed to solve the problem of localizability. In Fig.~\ref{sfig: Two receivers examples mono constellation}, one constellation is seen by the receivers, but only $7 < 7+C$ measurements are available. Thus, due to Lemma~\ref{lem: Minimum number}, the receivers in Fig.~\ref{sfig: Two receivers examples mono constellation} cannot be located. In Fig.~\ref{sfig: Two receivers examples bi constellation} with $C=2$ constellations, as discussed in the previous paragraph, the two receivers are localizable. We verify that, as stated in Lemma~\ref{lem: Minimum number}, the number of measurements satisfies: $9 \ge 7 + C$. Estimating the positions and the biases is another problem, which will be discussed in the following section.

\subsection{About the estimation of the positions}

Estimating the positions of the agents in a network is a difficult question. It has been the subject of numerous studies when only distance measurements are considered, see \eg{} \cite{aspnes2006theory} or \cite{patwari2005locating}. This paragraph highlights how GNSS rigidity can be used to design estimation algorithms. The algorithm proposed here is for illustrative purposes only. It is based on a perfect setting (no noise or bias), and its robustness and performance should obviously be the subject of specific studies.

The rigidity of the GNSS graph guarantees that the rigidity matrix has a maximum rank, which is not full because of the trivial motions. The rotation and the translation ambiguities are eliminated thanks to the satellites that have known positions. The clock bias ambiguity is removed by taking the first constellation as a reference. If we decompose the configuration as $\myvector{p} = \begin{pmatrix} \myvector{p}_a^{\intercal} & \myvector{p}_u^{\intercal} \end{pmatrix}^{\intercal}$, where $\myvector{p}_a$ groups all the known parameters (positions of the satellites, and biases of the first constellation), and $\myvector{p}_u$ groups all the unknown parameters, then the columns of $\mymatrix{R}_G(\Gamma, p)$ associated with $\myvector{p}_u$ have full rank. From \eqref{eq: Pseudorange rigidity matrix}, the rigidity matrix is proportional to the Jacobian of the measurement function $F_G$. The rigidity of the GNSS graph implies that the Jacobian of $F_G$ with respect to $\myvector{p}_u$ has full column rank and can be inverted. This property is crucial for the design of estimation algorithms.
For example, it allows to apply Newton's method, commonly used in the PVT algorithm to estimate the position of a single receiver \cite[Section 2.5]{kaplan2017understanding}. Starting from an initial configuration  $\myvector{p}^0 = \begin{pmatrix} \myvector{p}_a^{\intercal} & \myvector{p}_u^{0 \intercal} \end{pmatrix}^{\intercal}$, the vectors $\myvector{p}_u$ and $\myvector{p}$ are updated until convergence as:
\begin{subequations}
	\begin{align}
		\myvector{p}^{k+1}_u &\gets \myvector{p}^k_u - \left[\frac{\partial \myvector{F}_G(\Gamma, \myvector{p}^k)}{\partial \myvector{p}^k_u}\right]^+ (F_G(\myvector{p}^k) - \myvector{y}_m),\\
		\myvector{p}^{k+1} &\gets \begin{pmatrix} \myvector{p}_a^{\intercal} & \myvector{p}_u^{k+1\intercal} \end{pmatrix}^{\intercal},
	\end{align}
\end{subequations}
where $\left[\cdot\right]^+$ denotes the pseudo-inverse and $\myvector{y}_m$ is the vector of measurements. If correctly initialized, the algorithm converges to the positions. Once again, this algorithm is just for illustrative purposes and should be further studied in future work. The important point is that rigidity ensures that the pseudo-inverse can be formed which is essential for the algorithm.

\section{Pseudorange rigidity for formation control}\label{sec: Pseudorange rigidity for formation control}

Beyond GNSS, pseudorange rigidity has applications in other fields such as flight formation control. To make a group of UAVs flies in formation, a common strategy is to constrain some of the distances between the UAVs, see \eg{} \cite{olfati2002graph}. To maintain the formation, the structure formed by the agents must be rigid. Consequently, it requires at least $S_D(n,d)$ distance measurements. In practice, a distance measurement between UAVs is often carried out using the time-of-flight of a signal between the agents \cite{patwari2005locating}. As in GNSS, since the agents are not synchronized, the time-of-flight does not provided directly the distance between the agents.

If the clocks are modeled simply with a bias, as in \eqref{eq: Model of clock 1}, the time-of-flight gives the pseudorange between the agents. To suppress the bias, one can apply a symmetrical two-way ranging procedure by making two symmetrical pseudorange measurements and average them: $\delta_{u,v} = (\rho_{u,v} + \rho_{v,u})/2$. Therefore, maintaining the formation with this procedure requires constraining $2S_D(n,d)$ pseudoranges. If the agents were considered as a pseudorange framework instead of as a distance framework, it would require only $S_P(n,d)$ pseudorange constraints to maintain the formation. For large networks, as $2S_D(n,d) \sim_n 2nd$ and $S_P(n,d) \sim_n n(d+1)$, this second procedure reduces the number of measurements by up to $25\%$ in 2D and $33\%$ in 3D.

From an implementation point of view, pseudoranges have another interest. To control a formation the agents may be commanded to maintain only some of the constraints. Formation persistence \cite{hendrickx2007directed} studies the (distance) rigidity of graphs assuming that the constraints are maintained by only one agent, called the \emph{follower}. If the graph has some properties, the whole formation is preserved. This technique greatly simplifies the command. With the symmetrical two-way ranging procedure, an agent having several followers has to interact with every one of them to compute the distances. When the number of followers increases, the update rate necessarily decreases, which may induce a loss of precision. With the pseudorange approach in contrast, an agent having several followers may not interact with them: he could simply broadcast its position and bias, then, each follower could compute the pseudorange without any feedback. This approach allows significant scale up in the system as the number of followers would not be limited by the channel capacity. Persistence can be adapted to pseudoranges. Consider for example the rigid pseudorange graph in Figure~\ref{fig: Directed formation}. It requires feedback only between the first three agents. The rest of the agents can maintain the formation by maintaining only the pseudorange constraints pointing to them.
\begin{figure}[t]
	\centering
	{\input{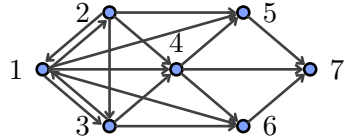}}
	\caption{Example of a rigid formation in $\R^2$.}
	\label{fig: Directed formation}
\end{figure}

If the clocks are modeled with a bias and a skew as in \eqref{eq: Model of clock 2}, the times-of-flight induce the constraints of JPC rigidity. In that situation, the pseudorange becomes:
	\begin{align*}
		\rho_{uv} &\triangleq c\left(t_{r(uv)}^v - t_{e(uv)}^u\right),\\
		&= w_v c t_{r(uv)} + \beta_v - w_u c  t_{e(uv)} - \beta_u, \\
		&= w_v \norm{x_u - x_v} +(w_v - w_u) c t_{e(uv)} + \beta_v - \beta_u,
	\end{align*}
	where $\beta_u \triangleq c\tau_u$.
	In \cite{wen2022clock}, it is assumed that $(i)$ the communication is bi-directional \cite[Assumption 1]{wen2022clock}, and $(ii)$ the agents send only one signal at a time $t_{e(u)}$ \cite[Assumption 2b]{wen2022clock}, \ie{} for all $v$, $t_{e(uv)} = t_{e(u)}$. In JPC rigidity, an edge constrains both the pseudoranges $\rho_{uv}$ and $\rho_{vu}$. We can imagine asymmetrical network in which, \eg{} some agents are only listening. The resulting framework would have a directed graph as a pseudorange graph.
A natural question is: Is asymmetrical JPC rigidity a generic property of the underlying undirected graph? The answer is no. For example, consider a fully connected graph of $4$ agents, such a graph is clearly symmetrical. Furthermore, according to \cite{wen2022clock}, it is JPC rigid in $\R^2$. Add a $5$th agent, we claim that if we add the arcs $(i, 5)$ for $i = 1, \dots, 4$, the graph is rigid, while if we add the arcs $(5, i)$ for $i = 1, \dots, 4$, it is not. Those two graphs are illustrated in Figure~\ref{fig: Example of JPC graphs}.
\begin{figure}
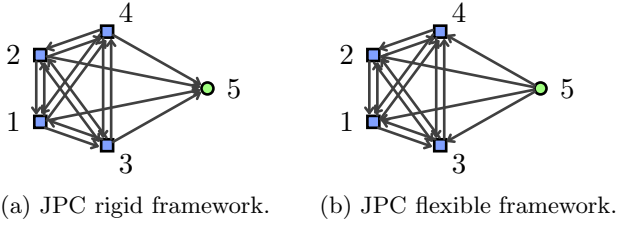

	\begin{subfigure}{.49\linewidth}
		\centering
		{\input{fig/jpc_rigid_graph.pgf}}
		\caption{JPC rigid framework.}
	\end{subfigure}
	\begin{subfigure}{.49\linewidth}
		\centering
		{\input{fig/jpc_flexible_graph.pgf}}
		\caption{JPC flexible framework.}
	\end{subfigure}
	\caption{Examples of asymmetrical JPC frameworks.}
	\label{fig: Example of JPC graphs}
\end{figure}
Indeed, the constraint induced by the arc $(i,5)$ is:
	\begin{subequations}
		\begin{equation}\label{eq: JPC Constraint i5}
			w_5 \left(\norm{x_i - x_5}+ t_{e(i)}\right) +\beta_i = \rho_{i5} + w_i  t_{e(i)} + \beta_i,
		\end{equation}
		while the constraint induced by the arc $(5,i)$ is:
		\begin{equation}\label{eq: JPC Constraint 5i}
			w_5 t_{e(5)} + \beta_5 = -\rho_{5i} + w_i(\norm{x_i - x_5} +  t_{e(5)}) + \beta_i.
		\end{equation}
	\end{subequations}
Constraint \eqref{eq: JPC Constraint i5} is similar to a GNSS constraint: with the first four agents set, the RHS is a constant. Generically, with $d+2$ equations the parameters of the $5$th agent are set, thus this graph is rigid. On the other hand, from \eqref{eq: JPC Constraint 5i}, neither $w_5$ nor $\beta_5$ can be identified but only the sum $w_5ct_{e(5)} + \beta_5$, thus this graph is flexible.
Asymmetrical JPC rigidity is therefore not a generic property of the underlying undirected graph. It is a more complex rigidity, since it uses a more complex clock model, and should be studied in future work.

\section{Conclusion}\label{sec: Concluding remarks}

This paper has introduced a new rigidity based on pseudoranges. It specificity is that the graph of constraints is directed. A complete characterization of generic rigidity of pseudorange graphs is provided based on distance rigidity. A pseudorange rigid graph is the combination of a distance rigid graph and a connected graph. Pseudorange rigidity has been adapted to answer the solvability of GNSS cooperative positioning problems. It was also highlighted how rigidity can be used to design algorithms to solve the positioning problem. Finally, the interest of pseudorange rigidity has been presented in formation control with a comparison with the recently introduced Joint Position-Clock rigidity.

Three possible future research directions have been suggested. The first is the study of global pseudorange rigidity. Once adapted to GNSS, it can provide the condition for the uniqueness of the solution. The second is the derivation of specific algorithms to estimate the positions of the agents. The third is the adaptation of pseudorange rigidity with more complex clock models.

\appendix

\section{Proof of Lemma~\ref{lem: Field extension}}\label{ap: Proof of algebraic lemmas}

Lemma~\ref{lem: Field extension} is the application of the following lemma.
\begin{lemma}\label{lem: Field extension multidimensional rigidity}
	Let $\K = \Q(X_1, \dots, X_N)$ be the field of fractions in $N$ variables with coefficients in $\Q$ and $(R_1, \dots, R_m)$ be a family of functions such that:
	\begin{enumerate}[label=\textbf{H.\arabic*}, ref=Hypothesis H.\arabic*]
		\item\label{it: H1} $\forall i \in \{1, \dots, m\}$, $R_i^2 \in \K$;
		\item\label{it: H2} $\forall i \in \{1, \dots, m\}$, $R_i^2 \notin \K^{(2)}$, with $\K^{(2)}=\{P^2 \mid P\in \K\}$;
		\item\label{it: H3} $\forall I \in \rP(\{1, \dots, m\})\setminus\{\emptyset\}$, $R_I = \prod_{i \in I} R_i \notin \K$.
	\end{enumerate}
	Then $\L = \K\left[R_1, \dots, R_m\right]$ is a field and $\L/\K$ is a field extension of order $2^m$.
\end{lemma}
\begin{proof}
	The proof is realized by induction over $m$. The property to prove is $\mathbf{P}(k)$: \guillemets{For any $R_1, \dots, R_k$ satisfying the three hypotheses, $\L = \K\left[R_1, \dots, R_k\right]$ is a field and $\L/\K$ is a field extension of order $2^k$.}
	
	\emph{Initialization.} For $k = 1$, let $R$ satisfy the three hypotheses. To prove that $\K[R]$ is a field, proving that every non-null element has an inverse is sufficient. Let $P \in \K[R]$, $P \ne 0$. Since $R^2 \in \K$, there exists $(A,B) \in \K^2$ with $(A,B) \ne (0,0)$ such that $P = A + BR$. If $B = 0$, $P = A \in \K$ therefore $P$ is invertible. If $B \neq 0$, using \ref{it: H2}, $R^2 \ne A^2/B^2$, therefore $A^2 - B^2 R^2 \neq 0$. Then, $(A - BR)/(A^2 - B^2R^2) \in \K[R]$ is the inverse of $P$. The extension is of order $2$ by \ref{it: H1} and \ref{it: H3}. Thus, $\mathbf{P}(1)$ is true.
	
	\emph{Induction step.} Assume $\mathbf{P}(k)$ for $k \ge 1$ and prove $\mathbf{P}(k+1)$. Let $R_1, \dots, R_{k+1}$ be $k+1$ functions satisfying the three hypotheses. We denote $\L_k = \K[R_1, \dots, R_k]$. First, let us prove that $\L_{k+1}$ is a field. Proving that $R_{k+1}^2 \notin \L_{k}^{(2)}$ is sufficient since then, with the same arguments as for the initialization every non-null element of $\L_{k+1}$ would have an inverse.
	Let us assume by contradiction that $R_{k+1}^2 \in \L_{k}^{(2)}$. By induction hypothesis, $\L_{k} = \L_{k-1}[R_k]$. Therefore, there exist $A, B \in \L_{k-1}$ such that:
	\begin{equation}
		R_{k+1}^2 = \left(A + B R_{k}\right)^2 = A^2 + B^2 R_{k}^2 + 2AB R_k
	\end{equation}
	If $AB \neq 0$, then $R_k \in \L_{k-1}$ which contradicts the induction hypothesis. Then, necessarily $A$ or $B$ is null. If $B = 0$, then $R_{k+1}^2 = A^2 \in \L_{k-1}^{(2)}$ This also contradicts the induction hypothesis when considering the $k$ functions $R_1,\dots, R_{k-1}, R_{k+1}$. Therefore $A = 0$. If $A = 0$ then, $R_{k+1}^2 = B^2 R_k^2$ and $(R_{k+1}R_k)^2 = (B R_k^2)^2 \in \L_{k-1}^{(2)}$. Similarly, this also contradicts the induction hypothesis when considering the $k$ functions $R_1, \dots, R_{k-1}, R_k R_{k+1}$ (which satisfies the three hypotheses). Theses contradictions give that $R_{k+1}^2 \notin \L_{k}^{(2)}$ and thus, $\L_{k+1}$ is a field.
	To prove the order, let us use the induction hypothesis:
	\begin{equation}
		[\L_{k+1}: \K] = [\L_{k+1}: \L_{k}][\L_{k}: \K] = 2^k [\L_{k+1}: \L_k]
	\end{equation}
	Since $R_{k+1} \notin \L_k$ and $R_{k+1}^2 \in \L_k$,  $[\L_{k+1}:\L_k] = 2$ and $[\L_{k+1}: \K] = 2^{k+1}$. $\mathbf{P}(k+1)$ is true.
\end{proof}

\begin{proof}[Proof of Lemma~\ref{lem: Field extension}]
	Let $E \subseteq \{uw \mid 1 \le u < v \le N\}$ be a set of edges and $m = \abs{E}$.
	The set of $m$ distance functions $D_{u,v}$ do satisfy the three conditions of Lemma~\ref{lem: Field extension multidimensional rigidity} when $d \ge 2$.
	
	Note however that when $d = 1$, the distance functions do not satisfy \ref{it: H2} of Lemma~\ref{lem: Field extension multidimensional rigidity}.
\end{proof}

\bibliographystyle{plain}
\bibliography{references.bib}

\begin{thebibliography}{10}

\bibitem{abel1991existence}
J.~S. Abel and J.~W. Chaffee.
\newblock Existence and uniqueness of gps solutions.
\newblock {\em IEEE Transactions on Aerospace and Electronic Systems},
  27(6):952--956, 1991.

\bibitem{asimow1978rigidity}
L.~Asimow and B.~Roth.
\newblock The rigidity of graphs.
\newblock {\em Transactions of the American Mathematical Society},
  245:279--289, 1978.

\bibitem{asimow1979rigidity}
L.~Asimow and B.~Roth.
\newblock The rigidity of graphs, ii.
\newblock {\em Journal of Mathematical Analysis and Applications},
  68(1):171--190, 1979.

\bibitem{aspnes2006theory}
J.~Aspnes, T.~Eren, D.~K. Goldenberg, A.~S. Morse, W.~Whiteley, Y.~R. Yang,
  B.~D.~O. Anderson, and R.~N. Belhumeur.
\newblock A theory of network localization.
\newblock {\em IEEE Transactions on Mobile Computing}, 5(12):1663--1678, 2006.

\bibitem{bancroft1985algebraic}
S.~Bancroft.
\newblock An algebraic solution of the gps equations.
\newblock {\em IEEE transactions on Aerospace and Electronic Systems},
  (1):56--59, 1985.

\bibitem{bollobas1998modern}
B.~Bollob{\'a}s.
\newblock {\em Modern graph theory}, volume 184.
\newblock Springer Science \& Business Media, 1998.

\bibitem{causa2018multi}
F.~Causa, A.~R. Vetrella, G.~Fasano, and D.~Accardo.
\newblock Multi-uav formation geometries for cooperative navigation in
  gnss-challenging environments.
\newblock In {\em IEEE/ION position, location and navigation symposium
  (PLANS)}, pages 775--785, Monterey, CA, USA, 2018.

\bibitem{chen2020angle}
L.~Chen, M.~Cao, and C.~Li.
\newblock Angle rigidity and its usage to stabilize multiagent formations in
  2-d.
\newblock {\em IEEE Transactions on Automatic Control}, 66(8):3667--3681, 2020.

\bibitem{connelly2005generic}
R.~Connelly.
\newblock Generic global rigidity.
\newblock {\em Discrete \& Computational Geometry}, 33(4):549, 2005.

\bibitem{fang2020angle}
X.~Fang, X.~Li, and L.~Xie.
\newblock Angle-displacement rigidity theory with application to distributed
  network localization.
\newblock {\em IEEE Transactions on Automatic Control}, 66(6):2574--2587, 2020.

\bibitem{frank1983current}
R.~L. Frank.
\newblock Current developments in {L}oran-{C}.
\newblock {\em Proceedings of the IEEE}, 71(10):1127--1139, 1983.

\bibitem{gluck1975almost}
H.~Gluck.
\newblock Almost all simply connected closed surfaces are rigid.
\newblock {\em Geometric Topology}, pages 225--239, 1975.

\bibitem{gortler2010characterizing}
S.~J. Gortler, A.~D. Healy, and D.~P. Thurston.
\newblock Characterizing generic global rigidity.
\newblock {\em American Journal of Mathematics}, 132(4):897--939, 2010.

\bibitem{hendrickson1992conditions}
B.~Hendrickson.
\newblock Conditions for unique graph realizations.
\newblock {\em SIAM journal on computing}, 21(1):65--84, 1992.

\bibitem{hendrickx2007directed}
J.~M. Hendrickx, B.~D.~O. Anderson, J.-C. Delvenne, and V.~D. Blondel.
\newblock Directed graphs for the analysis of rigidity and persistence in
  autonomous agent systems.
\newblock {\em International Journal of Robust and Nonlinear Control},
  17(10-11):960--981, 2007.

\bibitem{horn2012matrix}
R.~A. Horn and C.~R. Johnson.
\newblock {\em Matrix analysis}.
\newblock Cambridge university press, 2012.

\bibitem{jackson2007notes}
B.~Jackson.
\newblock Notes on the rigidity of graphs.
\newblock In {\em Levico Conference Notes}, volume~4, 2007.
\newblock Available at
  \url{https://webspace.maths.qmul.ac.uk/b.jackson/levicoFINAL.pdf}.

\bibitem{jacobs1997algorithm}
D.~J. Jacobs and B.~Hendrickson.
\newblock An algorithm for two-dimensional rigidity percolation: the pebble
  game.
\newblock {\em Journal of Computational Physics}, 137(2):346--365, 1997.

\bibitem{kaplan2017understanding}
E.~D. Kaplan and C.~Hegarty.
\newblock {\em Understanding GPS/GNSS: principles and applications}.
\newblock Artech house, 2017.

\bibitem{krause1987direct}
L.~O. Krause.
\newblock A direct solution to gps-type navigation equations.
\newblock {\em IEEE Transactions on Aerospace and Electronic Systems},
  (2):225--232, 1987.

\bibitem{minetto2019trade}
A.~Minetto, G.~Falco, and F.~Dovis.
\newblock On the trade-off between computational complexity and collaborative
  gnss hybridization.
\newblock In {\em IEEE 90th Vehicular Technology Conference (VTC2019-Fall)},
  pages 1--5, Honolulu, HI, USA, 2019.

\bibitem{nixon2018rigidity}
A.~Nixon, B.~Schulze, S.-I. Tanigawa, and W.~Whiteley.
\newblock Rigidity of frameworks on expanding spheres.
\newblock {\em SIAM Journal on Discrete Mathematics}, 32(4):2591--2611, 2018.

\bibitem{olfati2002graph}
R.~Olfati-Saber and R.~M. Murray.
\newblock Graph rigidity and distributed formation stabilization of
  multi-vehicle systems.
\newblock In {\em Proceedings of the 41st IEEE Conference on Decision and
  Control}, volume~3, pages 2965--2971, Las Vegas, NV, USA, 2002.

\bibitem{oxley2006matroid}
J.~G. Oxley.
\newblock {\em Matroid theory}, volume~3.
\newblock Oxford University Press, USA, 2006.

\bibitem{patwari2005locating}
N.~Patwari, J.~N. Ash, S.~Kyperountas, A.~O. Hero, R.~L. Moses, and N.~S.
  Correal.
\newblock Locating the nodes: cooperative localization in wireless sensor
  networks.
\newblock {\em IEEE Signal processing magazine}, 22(4):54--69, 2005.

\bibitem{roman2005field}
S.~Roman.
\newblock {\em Field theory}, volume 158.
\newblock Springer Science \& Business Media, 2005.

\bibitem{saliola2007some}
F.~V. Saliola and W.~Whiteley.
\newblock Some notes on the equivalence of first-order rigidity in various
  geometries.
\newblock {\em arXiv preprint arXiv:0709.3354}, 2007.

\bibitem{schmidt1972new}
R.~O. Schmidt.
\newblock A new approach to geometry of range difference location.
\newblock {\em IEEE Transactions on Aerospace and Electronic Systems},
  (6):821--835, 1972.

\bibitem{schulze2022frameworks}
B.~Schulze, H.~Serocold, and L.~Theran.
\newblock Frameworks with coordinated edge motions.
\newblock {\em SIAM Journal on Discrete Mathematics}, 36(4):2602--2618, 2022.

\bibitem{schulze2012coning}
B.~Schulze and W.~Whiteley.
\newblock Coning, symmetry and spherical frameworks.
\newblock {\em Discrete \& Computational Geometry}, 48(3):622--657, 2012.

\bibitem{wen2022clock}
R.~Wen, E.~Schoof, and A.~Chapman.
\newblock Clock rigidity and joint position-clock estimation in ultra-wideband
  sensor networks.
\newblock {\em IEEE Transactions on Control of Network Systems},
  10(3):1209--1221, 2023.

\bibitem{zhao2015bearing}
S.~Zhao and D.~Zelazo.
\newblock Bearing rigidity and almost global bearing-only formation
  stabilization.
\newblock {\em IEEE Transactions on Automatic Control}, 61(5):1255--1268, 2015.

\bibitem{zhao2019bearing}
S.~Zhao and D.~Zelazo.
\newblock Bearing rigidity theory and its applications for control and
  estimation of network systems: Life beyond distance rigidity.
\newblock {\em IEEE Control Systems Magazine}, 39(2):66--83, 2019.

\end{thebibliography}

\end{document}